\documentclass{Liebert_Author}
\usepackage{xcolor}
\usepackage[none]{hyphenat}
\usepackage{amsmath}
\usepackage{graphicx}  
\usepackage{verbatim} 
\usepackage{subfigure} 
\usepackage{amsfonts,amssymb,euscript,mathrsfs}
\usepackage{graphicx}
\usepackage{algorithmic}
\usepackage{todonotes}

\newtheorem{proposition}[theorem]{Proposition}

\title{Computing the Bounds of the Number of Reticulations in a Tree-Child Network That Displays a Set of Trees} 
\author{Yufeng Wu,$^{1\ast}$ Louxin Zhang ,$^{2}$\\
$^{1}$Department of Computer Science and Engineering, \\University of Connecticut,\\
Storrs, CT 06268, USA\\
$^{2}$ Department of Mathematics and\\
Center for Data Science and Machine Learning,\\ National University of Singapore, Singapore 119076\\
$^\ast$To whom correspondence should be addressed;\\
E-mail:  yufeng.wu@uconn.edu.
}

\begin{document} 

\maketitle 


\begin{abstract}
Phylogenetic network is an evolutionary model that uses a rooted directed acyclic graph (instead of a tree) to model an evolutionary history of species in which reticulate events (e.g., hybrid speciation or horizontal gene transfer) occurred. Tree-child network is a kind of phylogenetic network with structural constraints. Existing approaches for tree-child network reconstruction can be slow for large data. In this paper, we present several computational approaches for bounding from below the number of reticulations in a tree-child network that displays a given set of rooted binary phylogenetic trees. In addition, we also present some theoretical results on bounding from above the number of reticulations. Through simulation, we demonstrate that the new lower bounds on the reticulation number for tree-child networks can practically be computed for large tree data. The bounds can provide estimates of reticulation for relatively large data.
\end{abstract}

\section{Introduction}

%Background on tree-child networks.
%%%% COPIED FROM THE OTHER PAPER

Phylogenetic network is an emerging evolutionary model for several complex evolutionary processes, including recombination, hybrid speciation, horizontal gene transfer and other reticulate events \citep{Gusfield_book,Huson_book}.  On the high level, phylogenetic network is a leaf-labeled rooted acyclic digraph. Different from phylogenetic tree model, a phylogenetic network can have nodes (called reticulate nodes) with in-degrees of two or larger. The presence of reticulate nodes greatly complicates the application of phylogenetic networks. The number of possible phylogenetic networks even with a small number of reticulate nodes is very large \citep{fuchs_21_JEC}. 
A common computational task related to an evolutionary model is the inference of the model (tree or network) from data. A set of phylogenetic trees is a common data for phylogenetic inference. An established research problem on phylogenetic networks is inferring a phylogenetic network as the \emph{consensus} of multiple phylogenetic trees where the network satisfies certain optimality conditions \citep{elworth2019,gunawan_galled}. Each phylogenetic tree is somehow ``contained" (or ``displayed") in the network. 
The problem of inferring a phylogenetic network from a set of phylogenetic trees is called the network reconstruction problem (also called hybridization network problem in the literature). We refer to the recent surveys \citep{Steel_book,Zhang_18} for the mathematical relation between trees and networks.

The network reconstruction problem has been actively studied recently in computational biology. There are two types of approaches for this problem: unconstrained network reconstruction and constrained network reconstruction. Unconstrained network reconstruction \citep{ZW2012,WUMIRRN16,WURN10,WURNJCB13} aims to reconstructing a network without additional topological constraints. While such approaches infer more general networks, they are often slow and difficult to scale to large data. Constrained network reconstruction imposes some type of topological constraints on the inferred network. Such constraints simplify the network structure and often lead to more efficient algorithms. There are various kinds of constraints studied in the literature. One popular constraint is requiring simplified cycle structure in networks (e.g., so-called galled tree \citep{Gusfield_04,Wang_01}. 

Another topological constraint, the so-called tree-child property \citep{Cardona_09b}, has been studied actively recently. A phylogenetic network is tree-child if every non-leaf node has at least one child that is of in-degree one. This property implies that every non-leaf node is connected to some leaf through a path that is not affected by the removal of any reticulate edge (edge going into a reticulate node; see Figure~\ref{Fig1_examples}).
A main benefit of tree-child network is that it can have more complex structure than say galled trees, and is therefore potentially more applicable. While tree-child networks have complex structure, they can efficiently be enumerated and counted by a simple recurrence formula \citep{pons2021_SR,Zhang_19} and so may likely allow faster computation for other tasks. 
There is a parametric algorithm for determining whether a set of multiple trees can be displayed in a tree-child network simultaneously  \citep{van2022practical}.

\begin{figure}[t]
\label{Fig1_examples}
%\caption{XXX}
\centering
\includegraphics[width=0.80\textwidth]{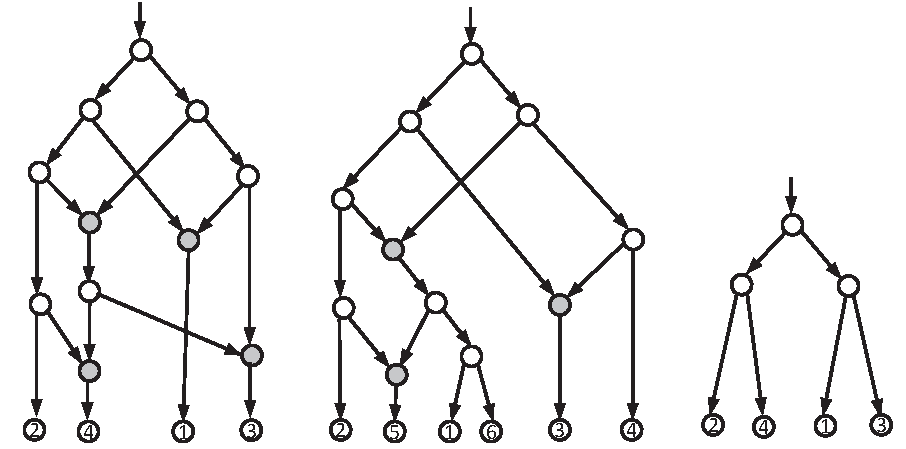}
\caption{An arbitrary phylogenetic network with four reticulate nodes (filled circles) on taxa 1, 2, 3, 4 (left), a tree-child network (middle) and a phylogenetic tree (right). Note: the tree on the right is displayed in the network to the left. }
\end{figure}

Given a phylogenetic network $N$, we say a phylogenetic tree $T$ (with the same set of taxa as $N$) is \textit{displayed} in $N$ if $T$ can be obtained by (i) first deleting all but one incoming edges at each reticulate node of $N$ (this leads to a tree), and then (ii) removing the degree-two nodes so that the resulting tree becomes a phylogenetic tree. As an example, in Figure \ref{Fig1_examples}, the tree on the right is displayed in the network on the left. Given a set of phylogenetic trees $\mathcal{T}$, we want to reconstruct a tree-child network such that it displays \textit{each} tree $T \in \mathcal{T}$ and its so-called reticulation number is the smallest among all such tree-child networks. Here, reticulation number is equal to the number of reticulate edges minus the number of reticulate nodes. The smallest reticulation number needed to display a set of trees $\mathcal{T}$ is called the tree-child reticulation number of $\mathcal{T}$ and is denoted as $\mathrm{TCR}_{\min}$.  Note that $\mathrm{TCR}_{\min}$ depends on $\mathcal{T}$. To simplify notations, we drop $\mathcal{T}$ from $\mathrm{TCR}_{\min}$ and the following lower bounds on $\mathrm{TCR}_{\min}$.
There exists no known polynomial-time algorithm for computing the exact $\mathrm{TCR}_{\min}$ for multiple trees.

Since computing the exact tree-child reticulation number $\mathrm{TCR}_{\min}$ of multiple trees is challenging, heuristics for estimating the range of $\mathrm{TCR}_{\min}$ have been developed. Existing heuristics aim at finding a tree-child network with the number of reticulation that is as close to $\mathrm{TCR}_{\min}$ as possible. At present, the best heuristics is ALTS \citep{Zhang_22}. ALTS can construct near-parsimonious tree-child networks for data that is infeasible for other existing methods. However, a main downside of ALTS is that it is a heuristic and so how close a network reconstructed by ALTS to the optimal one is unknown. Moreover, ALTS still cannot work on large data (say $50$ trees with $100$ taxa, and with relatively large number of reticulations).

We can view the network reconstruction heuristics as providing an \textit{upper bound} to the reticulation number. In order to gain more information on the reticulation number, a natural approach is computing a \textit{lower bound} on the reticulation number. Such lower bounds, if practically computable, can provide information on the range of the reticulation number. In some cases, if a lower bound matches the heuristically computed upper bound for some data, we can actually know the exact reticulation number \citep{WURN10}. Computing a tight lower bound on reticulation number, however, is not easy: to derive a lower bound one has to consider \emph{all} possible networks that display a set of trees $\mathcal{T}$; in contrast, computing an upper bound on reticulation number of $\mathcal{T}$  only requires one feasible network. For unconstrained networks with multiple trees, the only known non-trivial lower bound is the bound computed by PIRN \citep{WURN10}. While this bound performs well for relatively small data, it is computationally intensive to compute for large data. For tree-child networks, we are not aware of any published non-trivial lower bounds.

In this paper, we present several lower bounds on $\mathrm{TCR}_{\min}$. By simulation, we show that these lower bounds can be useful estimates of  $\mathrm{TCR}_{\min}$. In addition, we also present some theoretical results on upper bounds of $\mathrm{TCR}_{\min}$.

\subsection*{Background on tree-child network}
Throughout this paper, when we say network, we refer to tree-child network in which reticulate nodes can have two or more incoming reticulate edges (unless otherwise stated), which may not be binary. Edges in the network that are not reticulate edge are called tree edges. 
%Note that the network may not be binary. 
Trees are assumed to be rooted binary trees on the same taxa. 

 \textbf{The tree-child property}
A phylogenetic network is {\it tree-child} if every nonleaf node has at least one child that is a tree node.  In Figure~\ref{Fig1_examples}, the middle phylogenetic network is tree-child, whereas  the left network is not in which both the parent $u$ of the leaf 4 and the parent $v$ of the leaf 3 are reticulate and the node right above $u$ has $u$ and $v$ as its children.  One important property about tree-child network is that there is a directed path \textcolor{blue}{ consisting of } only tree edges from any node to some leaf (see e.g. \citet{Zhang_22}).

 \textbf{Network decomposition}
Consider a phylogenetic network $N$ with $k$ reticulate nodes. 
Let the root of $N$ be $r_0$ and the $k$ reticulate nodes be 
$r_1, r_2, \cdots, r_k$. 
For each $i$ from 0 to $k$, 
$r_i$ and its descendants that are connected to $r_i$ by a path consisting of only tree edges induces a subtree of $N$. Such $k+1$ subtrees are called the {\it tree components} of $N$ \citep{Gunawan_16_IC}. Note that the tree components are disjoint and the node set of $N$ is the union of the node sets of these tree components (see Figure \ref{Fig2_decomposition}).  Network decomposition is a powerful technique for studying the tree-child networks \citep{Cardona_20_JCSS,fuchs_21_JEC} and other network classes \citep{Gambette_15} (see \citet{Zhang_18} for a survey).

\begin{figure}[t]
%\caption{XXX}
\centering
\subfigure[Decompose into trees] % caption for subfigure a
{
    \label{Fig2_decomposition:trees}
%\scalebox{0.33}
{
  \includegraphics[width=1.52in]{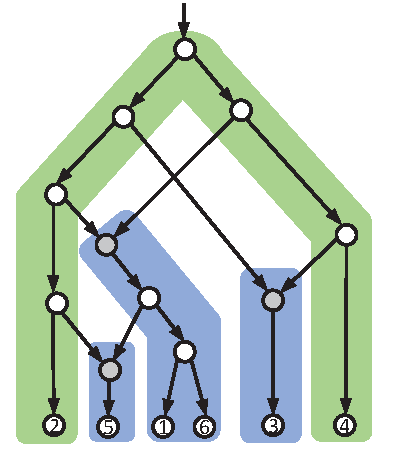}
}
}
\subfigure[Decompose into paths] % caption for subfigure a
{
    \label{Fig2_decomposition:path}
%\scalebox{0.33}
{
  \includegraphics[width=1.44in]{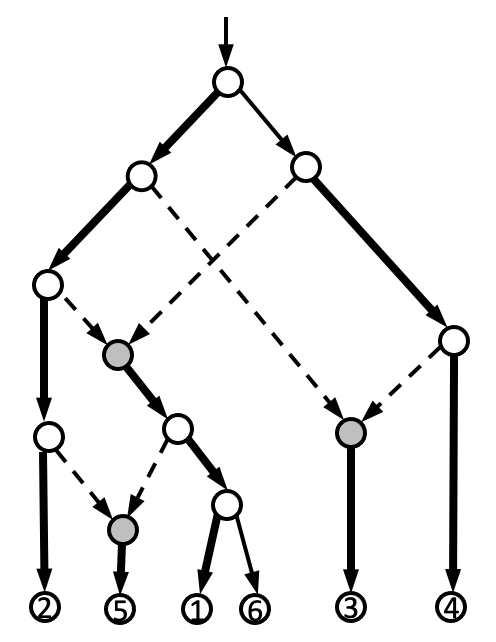}
}
}
\caption{Illustration of the decomposition of a phylogenetic network with $k$ reticulate nodes and $n$ leaves. In this example, $k=3$ and $n=6$. Part \ref{Fig2_decomposition:trees}: decompose into $k+1$ disjoint tree components. The tree component rooted at the network root is highlighted in green; other tree components  rooted at a recirculate node are in blue. Part \ref{Fig2_decomposition:path}: decompose into $n$ paths (each path appears in a tree component; ordered by the leaf labels). Reticulate edges: dashed lines. Edges in paths: thick lines. Tree edges not on paths: thin lines.
\label{Fig2_decomposition}}
\end{figure}

\textbf{Path decomposition}
The network decomposition for a tree-child network leads to a set of trees, where the trees are connected by reticulate edges. We can further decompose each tree component into \textit{paths} as follows. Suppose a tree component contains $p$ leaves and  these leaves are ordered in some way. We create a path for each leaf sequentially. Let $a$ be the current leaf. We create a path of edges from a node as close to the root of the the tree as possible, and down to $a$. We then remove \emph{all} edges starting at a path and ending at a different path. This procedure (called path decomposition) is illustrated in Figure \ref{Fig2_decomposition:path}, where the path creation follows the numerical order of the leaves. Note that path decomposition is a valid decomposition of a network $N$: each node in $N$ belongs to a unique path after decomposition. This is because only edges (not nodes) are removed during the above procedure. In addition,  path decomposition depends on the ordering on nodes: suppose we trace two paths backward to the root; when two paths meet, the path ending at an earlier node continues and the path to a later node ends. This implies that the ordering of leaves affect the outcome of the path decomposition. Moreover, a path starts at either a reticulate node or a tree node in path decomposition. At least one incoming edge is needed to connect the path to the rest of the network (unless the path starts from the root of the network).

\textbf{Displaying trees and path decomposition}
 When a tree $T$ is displayed in $N$, there are edges in $N$ that form a topologically equivalent tree (possibly with degree-two nodes) as $T$. Now,  when $N$ is decomposed into paths, to display $T$, we need to \textit{connect} the paths by using (either tree or reticulate) edges \textit{not} belonging to the paths.  Intuitively, tree edges connect the paths in a fixed way while reticulate edges lead to different topology of paths. That is, to display different trees, we need to connect the paths using different reticulate edges. This simple property is the foundation of the lower bounds we are to describe in Section~\ref{sectLBILP}. 

Recall that to display a tree, we need to make choices for each reticulate edge whether to keep or discard. This choice is called the \textit{display choice} for this reticulate edge.

%%%%%%%%%%%%%%%%%%%%%%%%%%%%%%%%%%%%%%%%%%%%%%%%%%%%%%%%

\section{Lower bounds on the tree-child reticulation number}
\label{sectLBILP}

In this section, we present several practically computable lower bounds for the tree-child reticulation number for displaying a set of trees. These bounds are derived based on the  decomposition of tree-child networks.

\subsection{$\mathrm{TCLB}_1$: a simple lower bound}

Recall that any tree-child network with $n$ taxa $\{1, 2, \ldots, n\}$ can be decomposed by path decomposition into $n$ simple paths (possibly in different ways), where each path starts with some network node and ends at a taxon. Now we consider a specific network $N$ and
a specific decomposition of $N$ into $n$ paths $P_i$ ($1 \leq i\leq n$). Each $P_i$ starts from nodes in the network and ends at taxon $i$. We say $P_i$ and $P_j$ are connected if some node within $P_i$ is connected by an edge to the start node of $P_j$ or vice versa. We define a binary variable $C_{i,j}$ to indicate whether or not $P_i$ and $P_j$ are connected for each pair of $i$ and $j$ such that $1 \le i < j \le n$. Note that $C_{i,j}$ is for a specific network $N$ and a specific path decomposition of $N$.  For an example, in the path decomposition in Figure \ref{Fig2_decomposition:path}, We have: 
$$C_{1,2}=C_{1,4}=C_{1,5}=C_{1,6}= C_{2,3}=C_{2,4}=C_{2,5}= C_{3,4}=1$$
and $C_{i, j}=0$ for other index pairs. 
Note that each of these $C_{i,j}=1$ corresponds to a specific (tree or reticulate) edge \textit{not} inside paths.
%
%
%We have the following simple observation.
 
 \begin{lemma}
 \label{lemmaLB0} 
 Let ${\cal T}$ be a set of trees.

{\rm (1)}  If $N$ is a network that displays $\cal T$, $C(N, D)=\sum_{1 \le i < j \le n} C^{(D)}_{i,j} -n +1 \le \mathrm{TCR}(N)$, 
where $C^{(D)}_{i,j}$ is obtained from the path decomposition constructed according to an arbitrary ordering $D$ on the taxa of the trees. Here, $\mathrm{TCR}(N)$ is the reticulation number of the network $N$. 

{\rm (2)} For any ordering $D$ on the taxa of $\cal T$, we have the following lower bound on $\mathrm{TCR}_{\min}
$: 
$$\mathrm{TCLB} \stackrel{\rm \Delta}{=} min_{N}  C(N, D) - n + 1\leq \mathrm{TCR}_{\min}.$$ 
 \end{lemma}
\begin{proof}
(1) Assume $\mathrm{TCR}(N)=k$. Let $D$ be an arbitrary ordering on the taxa $X$ of $\cal T$. 
The path decomposition constructed according to $D$ contains $n$ paths $P_1, P_2, \cdots, P_n$. Here, $C^{(D)}_{i,j} = 1$ if $P_i$ and $P_j$ are connected by an edge that goes from $P_i$ to $P_j$. That is, $P_i$ is obtained prior to $P_j$.

If $P_j$ starts from a tree node $u_j$ rather than the root of $N$, 
there is a unique path $P_i$ such that $C^{(D)}_{i,j}=1$ and $C^{(D)}_{i',j}=0$ for any $i'\neq i$. This implies that $\sum_{1\leq i\leq n} C^{(D)}_{i,j}=1=d_{in}(u_j)$.  Here, $d_{in}(u)$ is the in-degree of the node $u$.  If $u_j$ is the root of $N$, then $\sum_{1\leq i\leq n} C^{(D)}_{i,j} = d_{in}(u_j) = 0$: no path exists that is prior to $P_j$ with regard to $D$.

If $P_j$ starts from a reticulate node $u_j$, there are $d_{in}(u_j)$ reticulate edges entering $u_j$. Therefore, $P_j$ is connected with at most $d_{in}(u_j)$ paths. Therefore,
$\sum_{1\leq i \leq n}C^{(D)}_{i,j}\leq d_{in}(u_j).$

Summing these terms together, we obtain:
$$C(N, D)-n+1=\sum_{1\leq j\leq n} \sum_{1\leq i\leq n} C^{(D)}_{ij} -n+1 \leq \sum_{1\leq j\leq n} d_{in}(u_j) -n+1=\mathrm{TCR}_{\min}.$$

 (2)~  Let $O$ be a tree-child network with the smallest reticulate number $\mathrm{TCR}_{\min}$ that displays $\cal T$. For any ordering $D$ on the taxa, we have that 
   $C(O, D)\leq \mathrm{TCR}_{\min}$
and thus $\mathrm{TCLB}=min_{N}  C(N, D) \leq C(O, D)\leq \mathrm{TCR}_{\min}$. 
\end{proof}

%\todo{YW: I revised the above proof by stating C(i,j) given D is ordered. }

$\mathrm{TCLB}$ is a lower bound because it may underestimate $ \mathrm{TCR}_{\min}$ because $C_{i,j}$ are binary and there can be more than one edges connecting two paths in a path decomposition of the optimal network.
While Lemma~\ref{lemmaLB0} leads to a lower bound, $\mathrm{TCLB}$ is hard to compute because it needs to consider all possible networks $N$ that displays the given trees $\mathcal{T}$. We now show that we can practically compute a weaker bound $\mathrm{TCLB}_1$, which bounds from below $\mathrm{TCLB}$ and thus $\mathrm{TCR}_{\min}$.

We consider a binary tree $T \in \mathcal{T}$. 
(Our bounds can be generalized to non-binary trees.)
The following lemma illustrates one structural property of tree-child network when displaying a \emph{subtree} $T_1$ of $T$. 
Assume $T_1$ is rooted at node $v$.
Let $S(v)$ be the set of taxa under the node $v$. Since $T_1$ is also displayed in $N$, there exists some non-path edges (i.e., edges not on the paths in the path decomposition) which connects the paths, one path for each leaf in $S(v)$, that displays $T_1$. Let $v(N)$ be the node in $N$ that is the root of the displayed subtree in $N$. We say $T_1$ is displayed at node $v(N)$. 

\begin{lemma}
\label{lemmaSubtreeTC}
Let  $N$ be a tree-child network displaying $\mathcal{T}$ and let $T_1$ be a subtree rooted at $v$ of $T$ and be displayed at a node $v(N)$. Then, for any path decomposition, $v(N)$ is on some path. That is, we can always trace from a taxon from $S(v)$ upwards in $N$ and reach $v(N)$ by following only path edges for the path decomposition.
\end{lemma}

\begin{proof}
By the tree-child property, there is a leaf $a$ that can be reached from $v(N)$ following only tree edges. Thus, $v(N)$ and $a$ must be inside the same tree component (recall path component is obtained by further decomposition of some network decomposition into trees). Therefore, no matter how path composition is performed, there is always a leaf $a'$ where $v(N)$ is on the path ending at $a'$. 
%\hfill $\Box$
\end{proof}

\begin{lemma}
\label{lemmaSubtreeTC4}
\label{lemmaLB1}
Let $v$ be an internal node of $T \in \mathcal{T}$ with $v_1$ and $v_2$ as its children.
In a tree-child network $N$ displaying $\mathcal{T}$, for any path decomposition of $N$, 
\[    \sum_{i \in S(v_1)  }   \sum_{j \in S(v_2) }  C_{i,j}  \ge 1  \]
\end{lemma}

\begin{proof}
First $T \in \mathcal{T}$ is displayed in $N$. Then there exist edges of $N$ that connect the paths in a path decomposition to form $T$ (otherwise $T$ cannot be displayed in $N$). So suppose we trace these edges to locate the two subtrees rooted at $v_1$ and $v_2$. By Lemma \ref{lemmaSubtreeTC}, there are nodes $v_1(N)$ and $v_2(N)$  in $N$ where the two subtrees are displayed at, and are on some paths (denoted as $P$ and $P_j$ respectively). Here, $j$ is a taxon and $j \in S(v_2)$. When there are multiple such nodes for displaying an identical subtree, we choose the one that is closest to the root of $N$.

Now there is a node $v(N)$ in the network where the subtree of $T$ rooted at $v$ is displayed. Again by Lemma \ref{lemmaSubtreeTC}, $v(N)$ is on a path $P_i$ for some leaf $i$. This implies either $v_1(N)$ or $v_2(N)$ is on $P_i$ too. Without loss of generality, suppose $v_1(N)$ is on $P_i$. Then there must exist an edge between the path to $i$ and $P_j$ and $i \in S(v_1)$. This is because (i) there exists a path in $N$ from $v(N)$ to $v_2(N)$ that is taken to display $T$ in $N$; (ii) this path can have only a single edge; if not, then there exists at least a node $v_3$ \textit{not} on $P$ or $P_j$ (recall $v_2(N)$ is the one closest to the root among all choices for $v_2(N)$); (iii) let $v_3$ be on a decomposed path $P_k$ (which connects to a leaf $k$; but this violates the assumption that $v_1(N)$ and $v_2(N)$ display two subtrees of $v$.
This implies $C_{i,j} = 1$. We don't know which $i$ and $j$ for the network $N$. Nonetheless, there exists some $i \in S(v_1)$ and $j \in S(v_2)$ where $C_{i,j}=1$.  
%\hfill $\Box$
\end{proof}

Lemma~\ref{lemmaSubtreeTC4} leads to the following lower bound $\mathrm{TCLB}_1$.

\begin{proposition}
\label{propLB1}
Let $C_{i,j}$ be binary variables for $1 \le i < j \le n$. 
Let $\mathrm{TCLB}_1 = min (  \sum_{1 \le i < j \le n} C_{i,j} ) - n + 1 $ where $C_{i,j}$ satisfies the following constraint: for any internal node $v$ of a tree $T \in \mathcal{T}$ with two children $v_1$ and $v_2$, the condition stated in Lemma \ref{lemmaLB1} 
is satisfied.
Then $\mathrm{TCLB}_1$ is a lower bound on the tree-child reticulation number.
\end{proposition}

As an example, consider the tree on the right in Figure \ref{Fig1_examples}. We have the following constraints: $C_{2,4}  \ge 1, C_{1,3} \ge 1, C_{1,2}+C_{1,4} + C_{2,3} + C_{3,4} \ge 1$. When there are multiple trees, we create such constraints for each tree. $\mathrm{TCLB}_1$ takes the minimum over all choices of $C_{i,j}$ that satisfy all the constraints.

While we don't know how to efficiently compute $\mathrm{TCLB}_1$, 
it is straightforward to apply integer linear programming formulation (ILP) to compute $\mathrm{TCLB}_1$. Our experience in using ILP modelling shows that $\mathrm{TCLB}_1$ can usually be computed efficiently (in practice) even for large data: for 100 binary trees with 100 taxa, it usually takes less than one second even using a very basic ILP solver.

\subsection{$\mathrm{TCLB}_2$: a stronger lower bound}

We now present techniques to strengthen it to obtain a stronger lower bound called $\mathrm{TCLB}_2$. 
We start with a stronger version of Lemma~\ref{lemmaSubtreeTC4}.  
We need a special kind of path decomposition, called ``ordered path decomposition", of a network $N$. An ordered path decomposition is a path decomposition where its paths can be arranged in a total order, and \textit{all} reticulate and non-path tree edges are oriented in one direction relative to this total order. Such ordered path decomposition always exists. To see this,
 recall that $N$ is a digraph. Thus, all components of $N$ obtained by network decomposition can be arranged in a total order. Then we can obtain a tree decomposition by decomposing each component into paths. This leads to a tree decomposition where paths are linearly ordered from left to right and all reticulate edges and all non-path tree edges are oriented from left to right.  

We now consider an ordered path decomposition. We let $f(v)$ be the taxon in $S(v)$ that is ordered the first among all the taxa in $S(v$). That is, $f(v)$ is the taxon under node $v$ that is ordered the first among all the taxa (leaves) under $v$.

\begin{lemma}
\label{lemmaLB2}
Let $v_1$ and $v_2$ be the two children of node $v$ of some tree. Then,

\[   C_{ f(v_1),  f(v_2)}  = 1  \]
 
\end{lemma}

\begin{proof}
Recall the proof of Lemma \ref{lemmaSubtreeTC4}. When we trace the subtree rooted at $v_1$, the root of this subtree must be located within the simple path for $f(v_1)$. This is because the network is acyclic and the simple paths are ordered as in the specific path decomposition. Recall that all reticulate and non-path edges are oriented from left to right. So when we trace edges in a bottom up order (starting from leaves), we must reach the node (i.e. $f(v_1)$) that is ordered the first (i.e., the leftmost).
The situation for $v_2$ is similar. Thus, by the same reason as in Lemma \ref{lemmaLB1}, $P_{f(v_1)}$ and  $P_{f(v_2)}$ must be connected. 
%\hfill $\Box$
\end{proof}

Lemma \ref{lemmaLB2} leads to a stronger lower bound $\mathrm{TCLB}_2$. This is because if $C_{i,j}$ values satisfy the conditions in Lemma \ref{lemmaLB2}, they also satisfy the conditions in Lemma \ref{lemmaLB1}. 

Let $\mathcal{O}^*$ be the total order of the $n$ taxa in an ordered path decomposition.  We let $B( \mathcal{O^*}) = \sum_{i, j} C_{i,j}$, where $C_{i,j}=1$ if Lemma \ref{lemmaLB2} specifies which two taxa $i$ and $j$ must have $C_{i,j}=1$, when we consider \emph{all} internal nodes of each tree in $\mathcal{T}$. If $i$ and $j$ are not forced by Lemma \ref{lemmaLB2}, $C_{i,j}=0$. That is, $B( \mathcal{O^*})$ is fully decided if $\mathcal{O^*}$ is given. By Lemma \ref{lemmaLB2}, $B( \mathcal{O^*})-n+1$ is a lower bound on $\mathrm{TCR_{\min}}$.
One technical difficulty is that we don't know $\mathcal{O^*}$ for $N$. 
Nonetheless, we can derive a lower bound on $\mathrm{TCR}_{\min}$ by taking the \textit{minimum} over \textit{all} possible $\mathcal{O}$. Thus, we have the following observation. 

\begin{proposition}
\label{propLB2}

 $\mathrm{TCLB}_2 \stackrel{\rm \Delta}{=}  min_{\mathcal{O} }  (B(  \mathcal{O} ) - n + 1) $ is a lower bound on the tree-child reticulation number.
\end{proposition}

Naively, to compute $\mathrm{TCLB}_2$, we have to consider all possible total orders of the taxa. Enumerating all possible total orders of $n$ taxa is infeasible even for relatively small $n$ value. To develop a practically computable bound, we again apply ILP. 
We only provide a brief description of the ILP formulation.

We define binary variable $A_{i,j}$ for all $1 \le i, j \le n$ where $A_{i,j}=1$ if taxon $i$ is ordered earlier than taxon $j$. We need to enforce the ordering implied by $A_{i,j}$ is valid. That is, if taxon $i$ is earlier than $j$ and $j$ is earlier than $k$, then $i$ is earlier than $k$. This can be enforced in ILP as: for all $1 \le i, j, k \le n$ where $i,j,k$ are distinct:

\begin{align}
\label{equILP1}
A_{i,k}  + 1 \ge A_{i,j} + A_{j,k}  
\end{align}

We enforce the condition in Lemma \ref{lemmaLB2} by considering each taxon $i \in S(v_1)$ and taxon $j \in S(v_2)$:

\[  \sum_{p \in S(v_1), p \not=i}  A_{i,p} +  \sum_{q \in S(v_2), q \not=j}  A_{j,q}  \le C_{i,j} + |S(v_1)| + |S(v_2)| - 3  \]

This constraint enforces that $C_{i,j}=1$ if $i$ (respectively $j$) is the first taxon among $S(v_1)$ (respectively $S(v_2)$). Under these constraints, we use ILP to compute the $\mathrm{TCLB}_2$ by minimizing $\sum_{1 \le i < j \le n} C_{i,j}$. 

%%% Added Aug 28, 2023
The number of variables in this ILP formulation is $O(n^2)$, while the the number of constraints is $O(n^3)$ ($n$ is the number of taxa). The number of constraints (which is dominated by Equation \ref{equILP1}) can be large when $n$ increases. Note that since ILP formulation computes a lower bound, even if we skip some constraints in Equation \ref{equILP1}, the ILP still computes a lower bound. Empirical results appear to show that discarding some constraints often does not lead to a much weaker lower bound.

%%%%%%%%%%%%%%%%%%%%%%%%%%%%%%%%%%%%%%%%%%%%%%%%%%%%%%%%
\section{Bound in terms of cherries in the trees}

There is no known polynomial time algorithm for computing the lower bounds in Section \ref{sectLBILP}. A natural research question is developing good lower bounds that are polynomial time computable. 
In the following, we describe an analytical lower bound (called cherry bound) on tree-child reticulation number. Compared with the ILP-computed bounds in Section \ref{sectLBILP}, cherry bound is much easier to compute. However, experience shows that cherry bound tends to be weaker than the ILP-computed bounds. 
Cherry bound is expressed in terms of the number of distinct cherries in given trees $\mathcal{T} $. Here, a cherry is a 
two-leaf subtree in some $T \in \mathcal{T} $. We let $C$ be the number of \emph{distinct} cherries in the given trees $\mathcal{T}$.

Consider a tree-child network $N$ with $r$ reticulate nodes that displays $\mathcal{T}$, where $|\mathcal{T}| \ge 2$. Note that a reticulate node in $N$ has two or more incoming edges.  We let $n_R$ be the total number of reticulate edges of  $N$. That is, $n_R$ is equal to the sum of in-degrees of each reticulate node.  The reticulation number $R$ of  $N$ is equal to $n_R -r $. 

Now suppose we collapse common cherries in $\mathcal{T}$. Here, a common cherry is present in \emph{each} of the trees in $\mathcal{T}$. We collapse such common cherry into a single (new) taxon and repeat until there is no common cherry left. Note that this step is identical to common subtree collapsing,  which is a preprocessing step commonly practiced in phylogenetic network construction. 
 Collapsing identical subtrees in given set of trees is a common practice for computing $R_{\min}$ (see, e.g., \citet{Huson_book,Zhang_22}). So in the following, we assume there is no common cherry in $\mathcal{T}$.

Since cherry is a subtree of two leaves in $\mathcal{T}$, each cherry needs to be displayed in $N$ by obtaining a tree $T$ (through making display choices for reticulate edges) where $T$ displays this cherry. One can view the process of obtaining $T$ is traversing certain nodes of $N$. We have the following observation.

\begin{lemma}
\label{lemmaReducedCherry}
To obtain a cherry in $\mathcal{T}$,  we need to traverse 
either the tail or the head of some reticulate edge in $N$.  That is, displaying a cherry must depend on the choices we make about which reticulate edges to keep for displaying a tree. 
\end{lemma}

\begin{proof}

Suppose displaying a cherry in $N$ can be achieved by following a path that doesn't contain either the head or the tail of some reticulate edge. Then for any display choice (keep or discard) we make for reticulate edges, such path leading to the cherry that is always present. So, this cherry is a common cherry in $\mathcal{T}$, which contradicts our assumption of no common cherry. 
%\hfill $\Box$ 
\end{proof}

By Lemma \ref{lemmaReducedCherry}, each cherry in the given phylogenetic trees is related to the display choices in $N$. It is obvious that a cherry displayed in a tree must also be displayed in the network $N$. Therefore we consider the cherries displayed in the network. Suppose we add reticulate edges one by one to the network. Adding a reticulate edge can lead to new cherries to be displayed in the network. The more distinct cherries there are, the more reticulation is needed. We now make this more precise by establishing an upper bound on the number of distinct cherries that can be displayed by adding a single reticulate edge, which is an edge entering a reticulate node. Note that displaying a cherry can involve more than one reticulate edge. Suppose there are $R$ reticulations and so there are at least $2R$ reticulate edges. 

\begin{lemma}
\label{lemmaCherryMap}
 Selecting a reticulate edge $e_r$ to display a tree in a network $N$ can add at most $2$ distinct cherries.
\end{lemma}

\begin{proof}
Recall a cherry is a size-two subtree and is so displayed in the network $N$. To display a cherry in $N$, there are a set of tree or reticulate edges of $N$ that connect the two taxa of the cherry when displaying choices are made. We refer these edges as the cherry display of this cherry. We classify the cherries into two cases based on the types of edges in a cherry display.

\begin{quote}
   \begin{enumerate}
    \item[Type 1.] The cherry display contains at least one reticulate edge. That is, keeping a reticulate edge can only generate a type-1 cherry.
    \item[Type-2.] The cherry display contains only tree edges. That is, a type-2 cherry is only related to discarding (but not keeping) some reticulate edges. 
\end{enumerate}
\end{quote}

We now argue that keeping a reticulate edge $e_r$ can only generate at most one type-1 cherry and at most one type-2 cherry.
To see this, we first consider the case of keeping $e_r$. We call a taxon $a$ a \emph{tree}-taxon under an ancestor node $v$ if $a$ can be reached from $v$ by following a simple path with only tree edges, i.e. $a$ is a descendant of $v$ in the tree component containing $v$. Due to the property of tree-child network, at most one type-1 cherry can be obtained by keeping $e_r$: there must be only one tree-taxon $a$ below the destination of $e_r$, and one tree-taxon $b$ below the other child of the source of $e_r$, and keeping $e_r$ can only create a single distinct cherry $(a,b)$. Note that otherwise, no cherry can be formed by keeping $e_r$. If $e_r$ is kept, we have to remove its twin reticulate edges $e'_r$, this may display another cherry in the tree component containing the source node of $e'_r$, which is of type-2.

Therefore, we conclude that at most two distinct cherries can be associated with a reticulate edge. 
%\hfill $\Box$ 
\end{proof}

By Lemma \ref{lemmaReducedCherry}, each distinct cherry in $\mathcal{T}$ is associated with the display choices of some reticulate edge. By Lemma \ref{lemmaCherryMap}, one reticulate edge can lead to at most two distinct cherries. So $2n_R \ge C$. Note that reticulation number $R = n_R -r$ and $n_R \ge 2r$ (there are at least two reticulate edges per reticulate node). So, $R = n_R-r \ge \frac{n_R}{2}$. So,

\[  R \ge  \frac{n_R}{2}  \ge \frac{C}{4}  \]

\begin{proposition}
\label{propCherryBoundRed}
[Cherry bound two on reduced trees]  Let $C$ be the number of distinct cherries in a set of trees $\mathcal{T}$ which have no common cherries. We let $\mathrm{TCLB}_0 = \frac{C}{4}$ (called the cherry bound). Then $\mathrm{TCR}_{\min} \ge \mathrm{TCLB}_0$ (i.e., $\mathrm{TCLB}_0$ is a lower bound). 
\end{proposition}

Note that cherry bound is also valid when we restrict to binary tree-child networks.

\subsection{Another cherry-based lower bound}
We now present another efficiently computable lower bounds based the number of distinct cherries.

Let $N$ be a tree-child network displaying a set
${\cal T}$ of trees on $n$ taxa.  Let $C$ be the number of distinct cherries in the trees ${\cal T}$. Let $N$ have $r$ reticulations and $R$ denote the hybridization number of $N$. Then, the total in-degree  
of the $r$ reticulate nodes is 
$R+r$.   Then there are $(n-r-1)+(R+r)=n-1+R$ internal tree nodes.

Let $\ell_1$ and $\ell_2$ be two leaves of a cherry $Ch$ in a tree $T\in {\cal T}$. Then, $(\ell_1, p)\in {\cal E}(T)$ and $(\ell_2, p)\in {\cal E}(T)$ for some $p\in {\cal V}(T)$.  In the display of $Ch$ in $N$,  $p$ is mapped a tree node $\phi(p)$, $(\ell_1, p)$ and 
$(\ell_2, p)$ mapped to two node-disjoint paths $P_1$ and $P_2$. There are two possibilities:  
(i) $\ell_1$ and $\ell_2$ belong to one tree-node component and (ii) $\ell_1$ and $\ell_2$ are two different tree-node components. 

If $\ell_1$ and $\ell_2$ belong to 
a tree-node component, $\phi(p)$ is also in the same tree-node component. In this case, there are no other leaves below $\phi(p)$. Thus, $\phi(p)$ is uniquely determined by the two leaves. 

If $\ell_1$ and $\ell_2$ are in two different tree-node components, $\phi(p)$ and $\ell_1$ are in the same tree-node component, or
$\phi(p)$ and $\ell_2$ are in the same tree-node component.  Without loss of generality, we may assume the former holds. In this case, 
one child of $\phi(p)$ is the reticulate node on the top of the tree-node component $P$ containing $\ell_2$. Furthermore, $\ell_2$ is the only leaf in $P$. Therefore, 
$\phi(p)$ is also determined by the two leaves. 

In summary, we have proved the following property.

\begin{proposition}
 Let $v$ be a tree node of $N$ where a cherry in ${\cal T}$ is displayed at. Then, there are at most two leaves below $v$ in its tree-node component. If there are only two leaves $\ell_1$ and $\ell_2$ below $v$ in the tree-node component containing $v$, 
 only the cherry consisting of $\ell_1$ and $\ell_2$ can be displayed at $v$.  If there is only one leaf $\ell$ below $v$ in its tree-node component, then at most one cherry can be displayed at 
 $v$ under the condition that there is a unique leaf below the tree-node component rooted at the 
 reticulate child of $v$.
\end{proposition}

By the above proposition,  distinct cherries  are displayed at different tree nodes in $N$, Therefore, 
$$n-1+R\geq C.$$
or $$R\geq C-n+1.$$

Experiments show that while this bound is easily computable, it is often not as strong as $TCLB_0$ especially when $n$ is relatively large.

%%%%%%%%%%%

\section{On upper bounds on tree-child reticulation number}

So far we have focused on lower bounds on tree-child reticulation number $\mathrm{TCR}_{\min}$. A natural research question is developing sharp upper bounds on  $\mathrm{TCR}_{\min}$. Existing methods (e.g., \citet{Zhang_22}) can compute an upper bound for a given set of trees. However, little theoretical results are known for the computed bounds. In this section, we provide some theoretical results on upper $\mathrm{TCR}_{\min}$. 

We consider a set of $K$ trees $\mathcal{T}$. We consider a pair of trees $T_1, T_2 \in \mathcal{T}$. We let the tree-child reticulation number of $T_1$ and $T_2$. It is known that for two trees, tree-child reticulation number is equal to unconstrained reticulation number (which is also called hybridization number in the literature) \citep{LINZSEMPLE2019}. Hybridization number for two trees have been studied actively in the literature (see, e.g., \citet{BLSS07,WUHYBRIDDIST09,LINZSEMPLE2019}, among others). There are algorithms that can practically compute the hybridization number for two trees (see, for example, \citet{BLSS07,WUHYBRIDDIST09} among others). So we assume the hybridization number of $T_1$ and $T_2$ is known. We now describe an upper bound on  $\mathrm{TCR}_{\min}$ that uses the pairwise hybridization number. First, we need the following lemma, which is based on a result in \citet{WZ22} (also in \citet{Zhang_22}). The proof is based on a related proof in \citet{WZ22}.

\begin{lemma}
\label{lemmaTwoTreeTCN}
[\citet{WZ22}] For any two rooted binary phylogenetic trees $T_1$ and $T_2$ (over the same $n$ taxa), there exists a tree-child network $N$ that displays $T_1$ and $T_2$ with at most $n-2$ reticulations. Moreover, for any ordering of path components, there exists such an $N$ with this ordering of the path components in $N$. 
\end{lemma}

\begin{figure}[t!]
\centering
\includegraphics[scale=1.0]{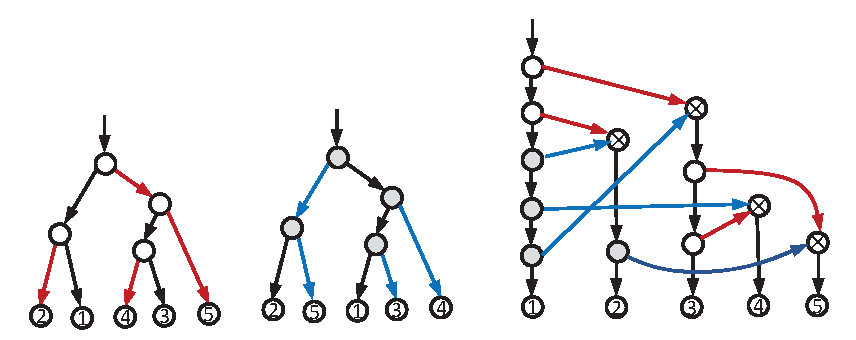}
\caption{Two trees $T_1$ (left), $T_2$ (middle) on four taxa (left) and the tree-child network (right), which displays the $T_1$ and $T_2$ simultaneously.
In $T_1$ (resp. $T_2$), the red (resp. blue) edges connect different paths of its decomposition. 
In each of the five tree components (vertical paths) of the network, the first node is its reticulate node; the unshaded vertices form the non-trivial paths appearing in the decomposition of $T_1$, while the shaded vertices form the non-trivial paths appearing in the decomposition of $T_2$. The red and blue reticulate edges correspond with the edges connecting different paths in the decomposition of $T_1$ and $T_2$, respectively.
\label{fig7_alg}
}
\end{figure}
\begin{proof}
Let $T_1$ and $T_2$ be two trees on $n$ taxa from $1$ to $n$. Without loss of generality, we order these $n$ taxa as $1, 2, \ldots, n$. 
We first decompose  $T_i$ into disjoints paths $P_{ik}$ ($1\leq k\leq n$) for $i=1, 2$ as follows. 
\begin{quote}
    1. $P_{i1}$ is the path consisting of the ancestors of leaf 1, together with the edges between them.
    
    2. For $k=2, \cdots, n$, $P_{ik}$ is the 
       the direct path consisting of the ancestors of leaf $k$ that do not belong to $\cup^{k-1}_{j=1}P_{ij}$ together with the edges between them.
\end{quote}
Let $p_i(k)$ be the parent of leaf $k$ in $T_i$. Note that $P_{i1}$ starts from the root of $T_i$ to $p_i(1)$. For $k\geq 2$, $P_{ik}$ is empty if $p_i(k)$ is in $\cup^{k-1}_{j=1}P_{ij}$ and non-empty otherwise. For example, for $T_1$ in Figure~\ref{fig7_alg}, $P_{11}$ is a 2-node path; $P_{12}$ is empty; $P_{13}$ is a 2-node path; and $P_{14}$ and $P_{15}$ are both empty.
We construct a tree-child network $N$ on $1$-$n$ with $n-1$ reticulate nodes (i.e. $n$ non-complex tree components) as follows. 

The first  component $Q_1$ of $N$ is obtained by connecting
$P_{11}$, $P_{21}$ and leaf 1 by edges (Figure~\ref{fig7_alg}). For $k>1$, 
the $k$-th component $Q_k$  is the concatenation of
a reticulate node $r_k$, $P_{1k}$, $P_{2k}$ and leaf $k$. Moreover, we connect the node that corresponds with the parent of 
the first node of $P_{ik}$ or $p_i(k)$ (if $P_{ik}$ is empty) to $r_k$ using (red or blue) edges for $i=1$  and $2$. In Figure~\ref{fig7_alg}, the red and blue reticulate edges are  added according to the path decomposition of $T_1$ and $T_2$, respectively.  

Since the edges not within a tree component are oriented from a node of  a tree component containing a leaf $i$ to the reticulate node of another tree component containing a leaf $j$ such that $i<j$, the resulting network is acyclic.  
It is easy to  see that the network is also tree-child. Moreover, 
 $T_1$ is obtained from $N$ if blue edges are removed and $T_2$ if red edges are removed.

The number of reticulations in $N$ is equal to the edges added in the algorithm above to connect $T_2$. Note that the first tree can be viewed as the ``tree part'' of $N$. Thus, the number of reticulations is $n-1$. Since there is no other taxa between taxa $1$ and $2$, we can only keep a \emph{single} edge from $1$ to $2$ (i.e., merge the two edges between $1$ and $2$ in Fig. \ref{fig7_alg}). This leads to $n-2$ reticulations.  
\end{proof} 

We now have the following upper bound on $\mathrm{TCR}_{\min}$.

\begin{theorem}
\label{theoremUpperboundTCN}
For $K$ phylogenetic trees $\mathcal{T}$ where there are two trees with $d$ as their pairwise hybridization number, then:

\[  \mathrm{TCR}_{\min} \le (K-2)(n-2) + d   \]
\end{theorem}

\begin{proof}
We first observe that Lemma \ref{lemmaTwoTreeTCN} can be naturally extended to $K$ trees. Intuitively, we can ``stack'' one tree after another using the constructive procedure in Lemma \ref{lemmaTwoTreeTCN}. Here, we use the same order of paths for the path decomposition of all trees in $\mathcal{T}$. This implies there is a tree-child network for the $K$ trees in $\mathcal{T}$ with at most $(K-1)(n-2)$ reticulations.  

Let $T_1$ and $T_2$ be two trees in $\mathcal{T}$ whose hybridization number is $d$. By \citet{LINZSEMPLE2019}, there exists a tree-child network $N$ with $d$ reticulations that displays $T_1$ and $T_2$. We consider a topological order $O$ of the path components of $N$. Now we build $N'$ that displays all trees in $\mathcal{T}$ by ``stacking'' each $T_i$ into $N'$ using the algorithm in Lemma \ref{lemmaTwoTreeTCN}. Here, we start with $T_1$ and $T_2$ as the first two trees to add into $N'$. Also, all trees in $\mathcal{T}$ are decomposed into paths with regarding to the order $O$. Therefore, we need $d$ reticulations to ``stack'' $T_2$ on top of $T_1$. By Lemma \ref{lemmaTwoTreeTCN}, ``stacking'' each additional $T_i$ ($3 \le i \le K$) needs at most $n-2$ reticulations. 
\end{proof}

%-----------------------

%%%%%%%%%%%

\section{Results}

We have implemented the lower bounds in the program PIRN, which is downloadable from \url{https://github.com/yufengwudcs/PIRN}. To compute the $\mathrm{TCLB}_1$ and $\mathrm{TCLB}_2$ bounds, PIRN uses \textit{GLPK}, an open-source ILP solver by default. While GLPK can practically compute $\mathrm{TCLB}_1$ for most data we tested, it becomes slow for computing $\mathrm{TCLB}_2$ for relatively large data. Our experience shows that $\mathrm{TCLB}_2$ can be practically computed using Gurobi, a more powerful ILP solver, even the data becomes relatively large. However, Gurobi is not open-source. In order to support Gurobi, PIRN outputs the ILP formulation in a file which can be loaded into Gurobi so that $\mathrm{TCLB}_2$ can be computed in an interactive way. 
The results we presented below were computed using Gurobi in this interactive approach.

\subsection{Simulation data}

To test the performance of lower bounds, we use the simulation data analyzed in \citet{Zhang_22}. The simulation data were generated using the approach first developed in \citet{WURN10}. Briefly, we first produced reticulate networks using a simulation scheme
similar to the well-known coalescent simulation backwards in time. At each step, there are two possible events: (a) lineage merging (which corresponds to speciation), and (b) lineage splitting (which corresponds to reticulation). The relative frequency of these two events (denoted as $r$) influences the level of reticulation in the simulated network: a larger $r$ will lead to more reticulation events in simulation. The following lists the simulation parameters.

%% Table 1
\begin{table}
\begin{center}
 \caption{\label{tableParam}\sl A list of parameters and their default values used in the simulation.}
{
       \begin{tabular}{ l | c | c }
       \hline
       Description & Symbol & Simulated values (default: boldface) \\ \hline \hline
       Number of taxa & $n$ & $\mathbf{10}, 20, 50$ \\ \hline
      Reticulation level & $r$ & $1.0, \mathbf{3.0}, 5.0$ \\ \hline
       Number of gene trees  &  $K$    &  $10, 50$ \\ \hline
  \hline
       \end{tabular}
}
\end{center}
\end{table}

We used the average over ten replicate data for each simulation settings. 
The following three lower bounds (all developed in this paper) were evaluated:
\begin{enumerate}
\item $\mathrm{TCLB}_0$: the cherry bound
\item $\mathrm{TCLB}_1$: the practically computable bound by ILP.
\item $\mathrm{TCLB}_2$: slower to compute by ILP but usually more accurate bound.
\end{enumerate}

In order to measure the accuracy of lower bounds, ideally we want to compare with the exact tree-child reticulation number. However, these methods tend to be slow for the data we tested. Therefore, we use the following two heuristic upper bounds instead as a rough estimate on tree-child reticulation number.

\begin{enumerate}
\item ALTS. This method calculates a heuristic upper bound on tree-child reticulation number.
\item PIRNs. Note: PIRNs outputs a unconstrained network. 
Since the output network may not be optimal, its reticulation number can occasionally be smaller than the computed lower bounds for tree-child reticulation number.  
But this is rare.
\end{enumerate}

We use the following statistics for benchmarking various methods. 

\begin{enumerate}
\item Average value of the (lower/upper) bounds.
\item For each lower bound, the average percentage of differences between a lower bound $LB$ and the ALTS bound $UB_a$:  $\frac{ UB_a-LB}{UB_a}$. 
\item Running time (in seconds).
\end{enumerate}

\begin{figure*}[ht]
\centering
\subfigure[Average lower/upper bound values] % caption for subfigure a
{
    \label{fig:res10tax:val}
%\scalebox{0.33}
{
  \includegraphics[width=1.6in]{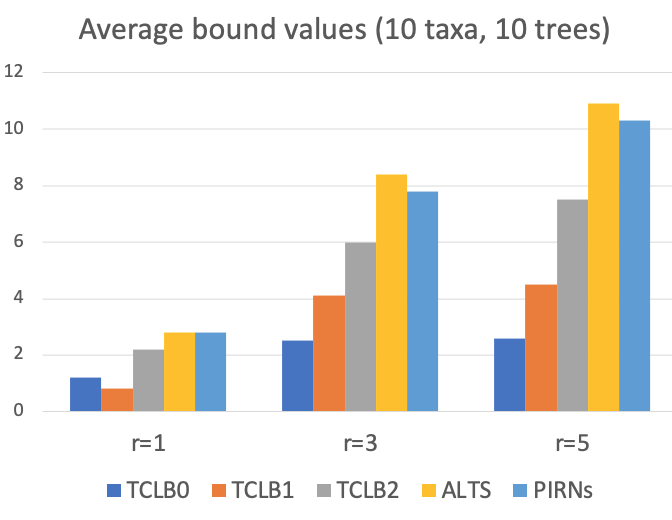}
}
}
\hspace*{5ex} 
\subfigure[Average gap between lower and the ALTS bound (normalized by ALTS) ] % caption for subfigure a
{
    \label{fig:res10tax:gap}
%\scalebox{0.33}
{
  \includegraphics[width=1.6in]{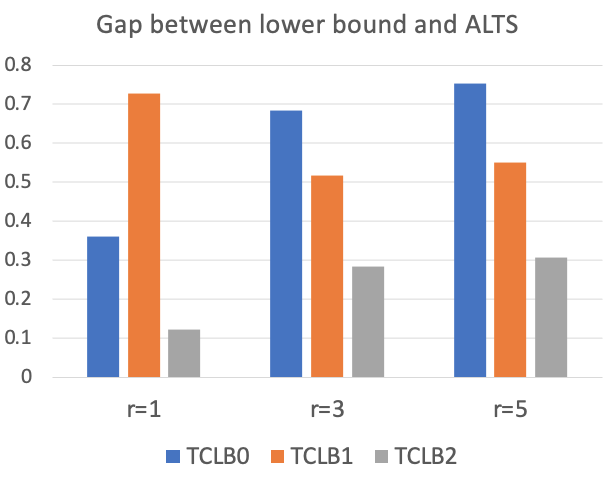}
}
}   
\caption{\label{fig:res10taxa}\sl Closeness of three lower bounds ($\mathrm{TCLB}_0$, $\mathrm{TCLB}_1$ and $\mathrm{TCLB}_2$) on 10 trees over 10 taxa under three reticulation levels $r=1, 3, 5$.  Two upper bounds, ALTS and PIRNs, are used for comparison.
Part \ref{fig:res10tax:val} shows the average lower/upper bound values. 
Part \ref{fig:res10tax:gap} shows the gap between each of the tree lower bounds and the ALTS bound (divided by the ALTS bounds).
}
%\vspace{-5pt}
\end{figure*}

Figure \ref{fig:res10taxa} shows the performance of the tree lower bounds, $\mathrm{TCLB}_0$, $\mathrm{TCLB}_1$ and $\mathrm{TCLB}_2$ on relatively small data (ten gene trees over ten taxa). Our results show that $\mathrm{TCLB}_2$ clearly outperforms the other two lower bounds in terms of accuracy. At lower reticulation level ($r=1$), the gap between $\mathrm{TCLB}_2$ and ALTS is only a little over $10\%$. At higher reticulation levels, the gap between $\mathrm{TCLB}_2$ and ALTS is larger but is still much smaller than the other two lower bounds. 
 Recall that ALTS is restricted to tree-child network while PIRNs works with unconstrained networks.

\begin{figure*}[ht]
\centering
\subfigure[Average lower/upper bound values for larger data] % caption for subfigure a
{
    \label{fig:resvartaxa:val}
%\scalebox{0.33}
{
  \includegraphics[width=1.9in]{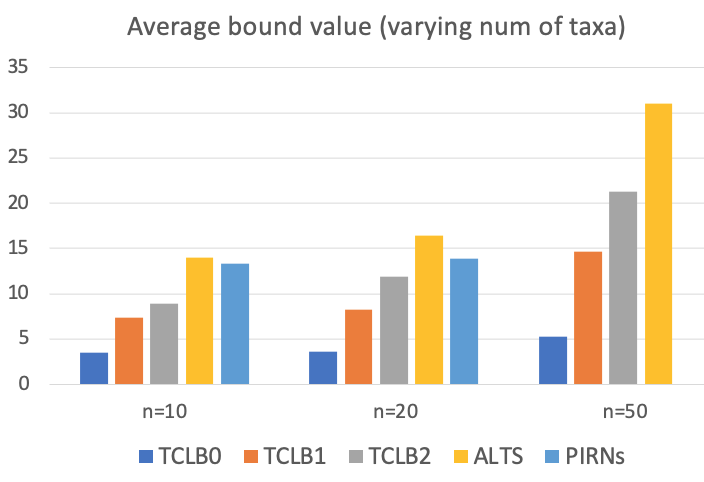}
}
}
\hspace*{5ex} 
\subfigure[Average gap between lower and the ALTS bound (normalized by ALTS) ] % caption for subfigure a
{
    \label{fig:resvartaxa:gap}
%\scalebox{0.33}
{
  \includegraphics[width=1.6in]{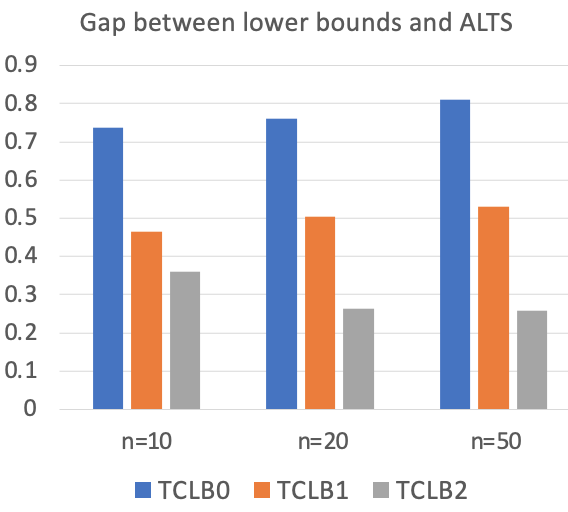}
}
}   
\caption{\label{fig:resvartaxa}\sl Performance of three lower bounds ($\mathrm{TCLB}_0$, $\mathrm{TCLB}_1$ and $\mathrm{TCLB}_2$) on larger data.  Reticulation level: $r=3$. $50$ gene trees. Vary the number of taxa ($n$): 10, 20 and 50.  Two upper bounds, ALTS and PIRNs, are used for comparison. PIRNs is too slow for n=50, and no result is given for this setting.
Part \ref{fig:resvartaxa:val} shows the average lower/upper bound values. 
Part \ref{fig:resvartaxa:gap} shows the gap between each of the tree lower bounds and the ALTS bound (divided by the ALTS bounds).
}
%\vspace{-5pt}
\end{figure*}

%We now evaluate the performance
We also examined the closeness
of the lower bounds on larger data. 
We simulated $50$ gene trees with varying number of taxa: $10, 20$ and $50$. Our results (Figure \ref{fig:resvartaxa}) show that $\mathrm{TCLB}_2$ still performs the best among the three lower bounds in term of the accuracy.

\noindent \textbf{Time to compute the bounds}
Figure \ref{fig:runtime} shows the running time to compute the bounds. We vary the reticulation levels (which may lead to networks with different number of reticulations), and also the number of taxa.
Our results show that computing $\mathrm{TCLB}_2$ takes longer time than the other two bounds. All lower bounds are faster to compute than the two upper bounds. 
ALTS is more efficient than PIRNs, while the ALTS bounds tend to be larger than the PIRNs bounds. PIRNs cannot be applied on large data (say $n=50$). ALTS also appears to be close to its practical range when $n=50$: there is one instance where ALTS failed to complete the computation by exhausting the memory in a Linux machine with 64 G memory).

\begin{figure*}[ht]
\centering
%\vspace{-20pt}
%\centerline{\scalebox{0.7}{\psfig{figure=exILS.eps,width=4.5in}}} 
\subfigure[Average run time for computing the lower/upper bounds with three different reticulation levels ] 
{
    \label{fig:runtime:time1}
%\scalebox{0.33}
{
  \includegraphics[width=1.8in]{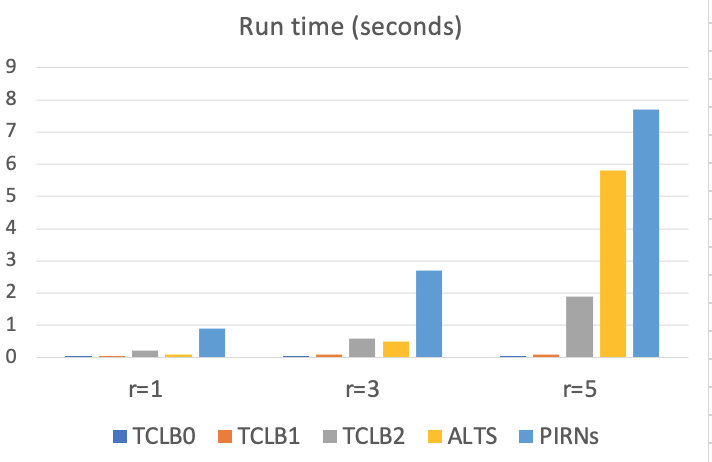}
}
}  
\hspace*{5ex} 
\subfigure[Average run time for computing the lower/upper bounds with varying numbers of taxa ] 
{
    \label{fig:runtime:time2}
%\scalebox{0.33}
{
  \includegraphics[width=1.8in]{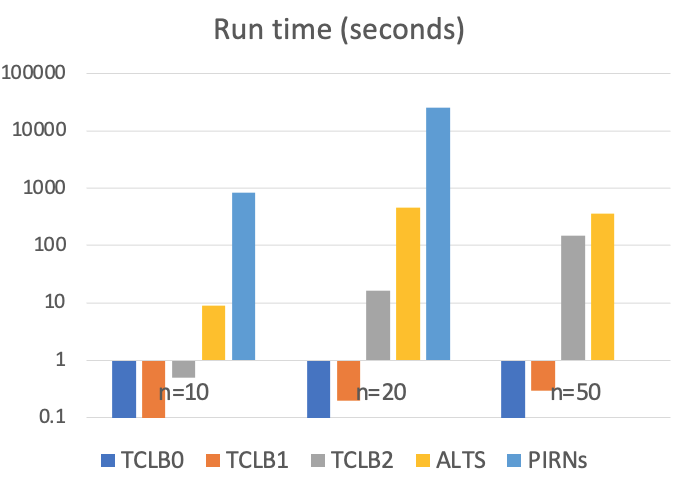}
}
}    \\
%\vspace{-10pt}
\caption{\label{fig:runtime}\sl Running time (in seconds) to compute three lower bounds ($\mathrm{TCLB}_0$, $\mathrm{TCLB}_1$ and $\mathrm{TCLB}_2$).  Two upper bounds, ALTS and PIRNs, are used for comparison.
Part \ref{fig:runtime:time1} shows the average run time (in seconds) for on 10 trees over 10 taxa under three reticulation levels $r=1, 3, 5$. 
Part \ref{fig:runtime:time2} shows the average run time (in seconds) for varying numbers of taxa ($n$): 10, 20 and 50 (reticulation level fixed at $r=3$ and $50$ gene trees). 
}
%\vspace{-5pt}
\end{figure*}

\subsubsection{More on large data}

$\mathrm{TCLB}_0$ can be easily computed for large data because it is based on simple properties of input trees and can be easily computed in polynomial time. While we don't have a polynomial time algorithm for computing $\mathrm{TCLB}_1$, our experience shows that $\mathrm{TCLB}_1$ can usually be easily computed even when only an open source ILP solver such as GLPK is used. This can be seen from Figure \ref{fig:runtime}.  
%\todo{Zhang:  Figure 5?}

$\mathrm{TCLB}_2$ can be practically computed using a state-of-the-art ILP solver such as Gurobi for moderately large data (e.g., $50$ gene trees with $50$ taxa). As an example, on a dataset with 50 trees (each with 50 taxa),  a lower bound of 16 is computed within a few seconds using Gurobi. The $\mathrm{TCLB}_1$ bound of $8$ blue was computed in a fraction of seconds even with an open source ILP solver. 
PIRNs took 10 hours to compute a unconstrained network with 20 reticulations.
ALTS took over 10 minutes to find a tree-child network with 23 reticulations. While the lower bound doesn't match the best upper bound, the lower bound can provide a range of the solution for large data. We note that Gurobi usually computes $\mathrm{TCLB}_2$ much faster than GLPK. Unless the data is small (say with 10 taxa or less), we recommend to use Gurobi.

To test its scalability, we simulated $50$ gene trees with $100$ taxa. $\mathrm{TCLB}_1$ can still be practically computed in less than one second even using GLPK.  $\mathrm{TCLB}_2$ can be computed using Gurobi, but in a long time. As an example, it took over 10 hours for obtaining $\mathrm{TCLB}_2 = 48$ on a dataset with $50$ simulated tree over $100$ taxa. In contrast, $\mathrm{TCLB}_1 = 37$ and $\mathrm{TCLB}_0=15$. Our experience shows that for very large data, the difference between $\mathrm{TCLB}_1$ and $\mathrm{TCLB}_2$ is not very large. Therefore, $\mathrm{TCLB}_1$ can provide a quick estimate on the reticulation number since it can be practically computed for large data, In fact, $\mathrm{TCLB}_1$ is perhaps the only practical method that can provide a reasonable strong estimate on reticulation for large data. We are not aware of any other existing approaches for estimating either a lower or upper bound that can be computed for the large simulated data we use here. Here, the large dataset mentioned here has $50$ gene trees, $100$ taxa and is simulated using reticulation parameter $r=3.0$ (which can lead to a tree-child reticulation number of over $40$).

\subsection{Real biological data}

To evaluate how well our bounds work for real biological data,
we test our methods on a grass dataset.
The dataset was originally from the Grass Phylogeny Working
Group \citet{GRASS01} and has been analyzed by a number of papers on phylogenetic networks. There are some variations in the exact form of data, depending on the preprocessing steps performed. The grass data we analyze here have five trees over $14$ taxa. Earlier analyses focus on calculating the so-called subtree prune and regraft distances between pairs of these trees \citet{BLSS07,WUSPR09,WUHYBRIDDIST09}. The first attempt for reconstructing phylogenetic network for all five trees is \citet{WURN10}. In \citet{WURN10}, the (unconstrained) reticulation number of these fives tree are known to be between $11$ (lower bound) and $13$ (upper bound). The upper bound was improved to $12$ by PIRNs \citep{WUMIRRN16}. Regarding to tree-child reticulation number, ALTS found a tree-child network with $13$ reticulations. No non-trivial lower bounds for tree-child reticulation number for these five grass trees are known before. 

We compute the three lower bounds on the five grass trees. The cherry bound $\mathrm{TCLB}_0$ is $2$, while the fast ILP bound $\mathrm{TCLB}_1$ is $3$. These two bounds can be calculated very fast but obviously the bounds are not very precise. It takes $75$ seconds to compute $\mathrm{TCLB}_2$ using Gurobi, which gives a lower bound of $11$. This matches the lower bound in \citet{WURN10}. Note that the lower bound in \citet{WURN10} is based on pairwise distances between the five trees, and takes much longer time to compute: when the number of tree increases, that bound becomes more difficult to compute. Although $\mathrm{TCLB}_2$ just provides the same bound as \citet{WURN10}, it is close to the currently best upper bound ($13$). Our results show that $\mathrm{TCLB}_2$ can indeed produce good estimates on tree-child reticulation number.

\section{Conclusion}

Our results show that the
lower bounds (especially $\mathrm{TCLB}_0$ and $TCLB_1$) are faster to compute than existing upper bounds (namely ALTS) on large data.  Our results show that there are trade-offs in accuracy and efficiency when computing lower bounds. The $\mathrm{TCLB}_2$ bound is the most accurate, but is also the slowest to compute. The simple cherry bound is very easy to calculate but usually is not very accurate. For large trees, the fast ILP-based $\mathrm{TCLB}_1$ bound may be a good choice to obtain quick estimate on tree-child reticulation number. We note that upper bound heuristics such as ALTS can construct a plausible phylogenetic network for the given gene trees, while lower bounds only provide a range of the reticulation number. Still, our lower bounds can provide quick estimate about the reticulation level of a set of phylogenetic trees for large data which is beyond the current feasibility range of existing upper bound methods.

Regarding to upper bounds, Theorem \ref{theoremUpperboundTCN} also gives an upper bound for hybridization number of $\mathcal{T}$, since a tree child network is a special case of hybridization network. However, reticulation number (with or without the tree-child condition) of three or more trees is still poorly understood. We are not aware of stronger upper bound than the bound in Theorem \ref{theoremUpperboundTCN} for three or more trees.
  
The tree-child network model often allows faster computation. The lower bounds on tree-child reticulation number are much faster to compute than lower bounds \citep{WURN10} on the general reticulation number.  
There are a number of open questions about lower bounds for tree-child reticulation number.  For example, is there a polynomial time algorithm for computing the $\mathrm{TCLB}_1$ bound?  Can one develop a new lower bound that has better (or similar) accuracy as $\mathrm{TCLB}_2$ and is faster to compute?

\subsubsection{Acknowledgments} Research is partly supported by U.S. NSF grants CCF-1718093 and IIS-1909425 (to YW) and Singapore MOE Tier 1 grant R-146-000-318-114 (to LZ). The work was started while YW was visiting the Institute for Mathematical
Sciences of National University of Singapore in April 2022, which was partly supported by
grant R-146-000-318-114.

\end{document}